%% file: main.tex
\documentclass[submission,copyright,creativecommons]{eptcs}
\usepackage[T1]{fontenc}

\usepackage{amsthm}
\usepackage{amsmath}
\usepackage{amssymb}
\usepackage{latexsym}
\usepackage{graphics}
\usepackage{epsfig,epsf,rotating}
\usepackage{subfigure}
\usepackage{subcaption}
\usepackage{epsf}
\usepackage{comment}

\usepackage{url}
\usepackage{doi}
\usepackage{float}
\usepackage{booktabs}
\usepackage{algorithm}
\usepackage{algpseudocode}
\algtext*{EndIf}
\algtext*{EndFor}
\algtext*{EndWhile}
\algtext*{EndFunction}
\algtext*{EndProcedure}
\usepackage{xcolor}
\usepackage{caption}
\usepackage{setspace}

\usepackage{placeins}

\theoremstyle{plain}
\newtheorem{theorem}{Theorem}[section]
\newtheorem{lemma}[theorem]{Lemma}

\newtheorem{corollary}[theorem]{Corollary}

\theoremstyle{remark}
\newtheorem{remark}[theorem]{Remark}

\input{macros}

\title{Rabin Games Revisited}
\author{Antonio Di Stasio
\institute{City St George's, University of London, UK}
\email{antonio.di-stasio@city.ac.uk}
\and
Moshe Y. Vardi
\institute{Rice University, Houston, USA}
\email{vardi@cs.rice.edu}
}

\newcommand{\titlerunning}{Rabin Games Revisited}
\newcommand{\authorrunning}{A. Di Stasio \& M. Y. Vardi}

\begin{document}
\maketitle
\begin{abstract}
 Rabin games have their roots in Michael Rabin's seminal work on automata over infinite trees and form a classical model for infinite-round games with $\omega$-regular winning conditions. They play a central role in automata theory and temporal synthesis: translating an LTL specification into an equivalent deterministic Rabin automaton reduces synthesis to solving a Rabin game. In practice, however, Rabin games are typically solved indirectly, via parity games. One possible reason is that many direct algorithms have a factorial dependence on the Rabin index $k$, the number of Rabin pairs. In this work, we revisit Rabin's 1972 recursive algorithm. The algorithm is implicit in Rabin's proof and intertwined with finite-state determinacy. First, we ``excavate'' it as an explicit arena-based solver using standard reachability- and B\"uchi-game solving as basic building blocks, thus avoiding the exponential dependence on the alphabet in Rabin's original presentation. Second, we show that memoization removes the factorial term from the recursion-tree analysis. Finally, we introduce a variation that recurses on the Rabin index. These results provide a modern analysis of Rabin's recursive method and a basis for studying direct Rabin-game solvers for temporal synthesis.
\end{abstract}
\input{introduction}

\input{1.1-related-work}
\input{1-preliminaires}
\input{2-rabin_algorithms}

\input{3-new_algorithm_induction}

\section{Conclusions}

We "excavated" here the Rabin-game-solving algorithm that is implicit in Rabin’s 1972 paper \cite{rabin1972automata}. We obtained an explicit solver based on standard reachability- and Büchi-game solving. This perspective yields three main observations. First, the algorithm can be presented without the exponential dependence on the alphabet that appears in Rabin’s original formulation. Second, memoization removes the factorial term arising from the recursion-tree analysis by ensuring that each normalized subproblem is solved only once. Third, the recursion can be parameterized by the Rabin index rather than the Rabin number, yielding an alternative solver whose complexity depends on the number of Rabin pairs.  
We leave it to future work to compare the empirical performance of the two algorithms studied here. 

It also remains an interesting direction for future work to investigate symbolic implementations of these recursive procedures and to evaluate them experimentally. In particular, Rabin conditions arise naturally in automata-theoretic approaches to reactive synthesis, where an LTL specification is translated into a deterministic Rabin automaton. Nevertheless, current synthesis pipelines typically proceed by reducing the resulting Rabin game to a parity game. It would therefore be valuable to study the structure of the Rabin conditions that arise in practice and to assess whether direct Rabin-game solving can provide an effective algorithmic basis for temporal synthesis.

\newpage

\bibliographystyle{eptcs}
\bibliography{references}

\end{document}

%% file: macros.tex
\newcommand{\Nat}{{\rm I\kern-.23em N}}


%% file: introduction.tex
\section{Introduction}

Rabin games originate in a central early development of automata theory. In his proof of the decidability of the monadic second-order theory of infinite trees, Rabin introduced acceptance conditions that combine recurrence and eventual avoidance requirements~\cite{rabin69}. This decidability result, known as \emph{Rabin's Tree Theorem}, had a major impact on logic, automata theory, and infinite-duration games. Rabin's original decision procedure was nonelementary, and much subsequent work was devoted to obtaining sharper algorithmic bounds. In particular, Hossley and Rackoff later gave an alternative elementary decidability proof for the emptiness problem of Rabin automata on infinite trees ~\cite{HossleyR72}.

Rabin gave a more direct treatment of Rabin automata on infinite trees and Church's synthesis problem~\cite{rabin1972automata}. Although presented in a different form, his proof contains an implicit recursive procedure for solving Rabin games. A Rabin game is played on a finite arena by two players. Its winning condition is given by pairs $(L_i,R_i)$, and Player~0 wins a play if, for some pair, states from $L_i$ are visited infinitely often while states from $R_i$ are visited only finitely often \cite{buhrke1996strategy}. Thus, Rabin conditions express the simultaneous need to visit good sets infinitely often and to eventually avoid bad sets.

The algorithmic problem of solving Rabin games is to compute the states from which Player~0 has a winning strategy. Emerson and Jutla showed that this problem is NP-complete~\cite{EJ99}. Beyond its theoretical importance, the problem gained renewed relevance through the development of temporal synthesis initiated by Pnueli and Rosner~\cite{pnueli1989synthesis}. In temporal synthesis, the goal is to automatically construct a system that satisfies a temporal specification against all possible behaviours of the environment. This problem can be reduced to solving Rabin games, with the system and the environment viewed as the two players, by converting an LTL formula to an equivalent deterministic Rabin automaton on infinite words~\cite{pnueli1989synthesis}. 

Although Rabin conditions arise naturally in the automata-theoretic solution of reactive synthesis~\cite{pnueli1989synthesis}, modern synthesis pipelines often avoid solving Rabin games directly. Instead, they reduce the problem to parity games and then apply parity-game solvers ~\cite{KMSZ18,LMS20}. This approach is well established and has led to efficient tools \cite{MSL18, KMPZ25}, but it may also obscure the structure of the original Rabin condition. Indeed, it is explicitly observed that current Rabin-game solvers are not sufficiently efficient in practice \cite{KMSZ18}. This suggests that direct Rabin-game solving remains an important underexplored direction.

One possible explanation is that most known direct algorithms for Rabin games have a complexity bound containing a factorial dependence on the Rabin index (the number of Rabin pairs) $k$. This leaves open whether a different algorithmic view of Rabin games can yield a different complexity tradeoff. In particular, Rabin's 1972 algorithm has received comparatively little direct analysis. One reason is that the algorithm is not presented as a standalone solver: it is intertwined with the proof that whenever Player~0 wins a Rabin game, Player~0 also has a finite-state winning strategy. In addition, the original formulation yields an exponential dependence on the alphabet, to simplify the description of the algorithm.

The goal of this paper is to contribute to the theory of Rabin games by first revisiting this original recursive construction. We start by ``excavating'' the algorithm implicit in Rabin's 1972 proof and presenting it as an explicit solver for Rabin games. The resulting algorithm works directly with the Rabin condition. Our first observation is that the exponential dependence on the alphabet in Rabin's original presentation can be eliminated. By formulating the construction directly over arenas and using standard reachability- and B\"uchi-game solving, the algorithm can be analyzed in terms of the size of the game representation rather than the size of the explicit alphabet expansion.

Our second observation concerns memoization. A naive recursive implementation may solve the same normalized subproblem many times, leading to a factorial-size recursion tree. However, if recursive calls are memoized, then each distinct normalized subproblem is solved only once. This removes the explicit factorial term from the recursion-tree analysis, at the cost of an exponential dependence on the space of state-set configurations.

Finally, we show that the recursion need not follow Rabin's original parameter, which is the Rabin number, namely the number of states appearing in the right-hand sides of Rabin pairs. We offer a novel variation that recurses on the Rabin index, eliminating one Rabin pair at a time. This is significant for temporal synthesis, where the Rabin index can be much smaller than the Rabin number, and in particular may be logarithmic in it.

Overall, this paper provides a modern presentation and analysis of Rabin's recursive method, together with memoized state-based and index-based solvers. The purpose is not to claim dominance over existing parity-based or progress-measure algorithms, but to clarify a different algorithmic route for Rabin games. We believe this opens the door to future experimental work on whether direct Rabin-game solving can serve as an algorithmic basis for temporal synthesis.

%% file: 1.1-related-work.tex
\section{Related work}
A  classical algorithmic line for Rabin games is based on rankings and
progress measures. Piterman and Pnueli introduced direct Rabin and Streett rankings, obtaining an explicit ranking algorithm for Rabin games with time complexity $O(mn^{k+1}kk!)$ and space $O(nk)$, where $n$ is the number of states, $m$ the number of edges, and $k$ the Rabin index (number of Rabin pairs)~\cite{PitermanP06}. Their recursive fixpoint characterization also yields a symbolic version with stored intermediate values,  with time $O(n^{k+1}k!)$ and space $O(n^{k+1}k!)$.

More recently, Majumdar \emph{et al.} \cite{MST24} developed progress-measure
algorithms based on colourful universal trees, with running time
$\widetilde{O}(mn(k!)^{1+o(1)})$ and polynomial space.
These algorithms are polynomial in the arena size for fixed $k$, but still
retain a factorial or near-factorial dependence on the Rabin index.
Recent lower bounds show that this dependence is essentially unavoidable for
algorithms that remain polynomial in the arena size: assuming the Exponential Time Hypothesis (ETH), Rabin games
with $k$ pairs cannot be solved in time $2^{o(k\log k)}n^{O(1)}$
~\cite{casares2024simple}. 

A related line of work studies regular games through uniform
recursive and dynamic-programming methods.  In~\cite{LKX26}, a unified framework is developed, which applies to all regular
games, including Rabin and Streett games, via two distinct approaches. Their
\emph{recursive} solver transforms Rabin conditions into equivalent colored
Muller conditions and applies a recursive colored-Muller algorithm, running in time $O(|V|!\,|V|\,(|E|+k|V|))$ and space $O(|\mathcal{G}|+|V|^2)$.
Their \emph{dynamic-programming} solver transforms the Rabin condition into 
an equivalent Muller condition and applies a bottom-up dynamic-programming 
algorithm over all $2^{|V|}$ subsets, running in time $O((k|V|+2^{|V|}|E|)\,|V|)$ and space $O(|\mathcal{G}|+2^{|V|}\,|V|)$. The paper describes the latter as the first exponential-time method for Rabin games, arguing that all prior direct algorithms were super-exponential. Both algorithms are most competitive when $k$ is large relative to $|V|$ --- specifically when $k \in [|V|, 4^{|V|}]$. 

Our work is also related to symbolic algorithms for Emerson--Lei games.
Hausmann, Lehaut, and Piterman give a symbolic solution based on the Zielonka
tree of the Emerson--Lei acceptance formula~\cite{hausmann2024symbolic}.
Rabin objectives are expressible as Emerson--Lei objectives, but our construction does not proceed by building the Zielonka tree. Instead, it follows a Rabin-specific reduction: Rabin states or right-hand sides of Rabin pairs are declared unsafe, deleted from the condition, normalized, and already-winning states are folded into the reachability target. Thus the approaches solve overlapping classes of games, but their recursive structure and complexity tradeoffs are different.

%% file: 1-preliminaires.tex
\section{Preliminaries}
\subsection{Arenas and Games}

An \emph{arena} is a tuple $A=\langle V,v_0, \Sigma^0,\Sigma^1, \delta \rangle$, where $V$ is a finite set of nodes or \emph{states}, $v_0\in V$ is the initial state, $\Sigma^i$ is the set of actions of Player~$i$, for $i=0,1$, and $\delta: V\times\Sigma^0\times\Sigma^1 \rightarrow V$ is a total transition function, yielding a new state after actions from $\sigma^0$ and $\sigma^1$ are chosen by Player 0 and Player 1, respectively.
A \emph{play} in $A$ is an infinite sequence $\pi=(v_0,\sigma^0_0,\sigma^1_0), (v_1,\sigma^0_1,\sigma^1_1)\ldots \in (V\times\Sigma^0\times\Sigma^1)^\omega$, starting at $v_0$,  and such that for every $i\ge 0$, $v_{i+1} = \delta(v_i,\sigma^0_i,\sigma^1_i)$. Intuitively, at state $v_i$, Player 0 first chooses an action $\sigma_i^0\in\Sigma^0$, then Player 1 chooses an action $\sigma_i^1\in\Sigma^1$. The play then moves to $v_{i+1}$. 

We assume throughout that the initial state $v_0$ is a special state, in the sense that it is not in the range of $\delta$; that is, it is \emph{non-recurrent}. This assumption is without loss of generality, since any arena can be transformed into an equivalent one satisfying it by adding at most one fresh initial state. Formally, given an arena $A=\langle V,v_0,\Sigma^0,\Sigma^1,\delta\rangle$, we can construct an equivalent arena
$A'=\langle V\cup\{v_0'\},v_0',\Sigma_0,\Sigma_1,\delta'\rangle$, where $v_0'\notin  V$ is a new state, $\delta'(v_0',\sigma_0,\sigma_1) = \delta(v_0,\sigma_0,\sigma_1)$, and $ \delta'(v,\sigma_0,\sigma_1) = \delta(v,\sigma_0,\sigma_1)$ for all $v\in V$. 
To simulate a game starting from an arbitrary state $v$, we will make $v_0$ \emph{mimic} the state $v$. Given a state $v\in V$, we write $A[v_0 \mapsto v]$ for the arena obtained from $A$ by replacing the transition function at $v_0$ with the transition function of $v$, i.e., $\delta'(v_0,\sigma_0,\sigma_1) = \delta(v,\sigma_0,\sigma_1)$  for all $\sigma_0\in\Sigma_0$ and $\sigma_1\in\Sigma_1$, while $\delta'$ agrees with $\delta$ on all other states.


A \emph{game} $G=\langle A,\alpha\rangle$ is a pair consisting of an arena $A$ and a \emph{winning condition} $\alpha\subseteq (V\times\Sigma^0\times\Sigma^1)^\omega$, denoting the plays that are \emph{winning} for Player 0. A strategy for Player 0 is a function $s: (\Sigma^1)^* \rightarrow \Sigma^0$.  A play $\pi=(v_0,\sigma^0_0,\sigma^1_0), (v_1,\sigma^0_1,\sigma^1_1)\ldots$ is \emph{consistent} with the strategy $s$ if $v_{i+1}=\delta(v_i,\sigma_i^0,\sigma_i^1)$, where $\sigma_i^0=s(\sigma_0^1,\ldots,\sigma_{i-1}^1)$. A strategy $s$ is winning for Player~0 in $G$ if every play consistent with $s$ belongs to $\alpha$. 
Player 0 \emph{wins the game} $G=\langle A,\alpha\rangle$ if it has a \emph{winning strategy}, that is, a strategy $s$ such that all plays that follow $s$ are winning plays. 

Given an arena $A=\langle V,v_0, \Sigma^0,\Sigma^1, \delta \rangle$ and a set $W \subseteq V$ of states, we define $Force_0(W)$ to be the set of states from which Player 0 can force a move into $W$. 
Formally,
\[
Force_0(W) = \left\{ v\in V \;\middle|\; \exists \sigma_0\in\Sigma^0.\forall \sigma_1\in\Sigma^1,\
\delta(v,\sigma_0,\sigma_1)\in W
\right\}.
\]

\subsection{Reachability, B\"uchi, and Reachability-B\"uchi Games}

In a \emph{reachability game} $G=\langle A,F\rangle$ the winning condition $\alpha$ is defined by a subset $F\subseteq V$ of target states. A play $\pi$ is winning for Player 0 if some state from $F$ \emph{occurs} in $\pi$. Equivalently, Player 0 wins $G$ if it can force the play to eventually visit $F$. We assume that the initial state $v_0$ is not in $F$.

To solve a reachability game $G=\langle A,F \rangle$ we iteratively compute the set of states from which Player~0 can force a visit to $F$. Define $W^r_0 := F$, and for every $i\ge 0$, $W^r_{i+1} := W^r_i \cup Force_0(W^r_i)$.  Since the sequence $W^r_0 \subseteq W^r_1 \subseteq W^r_2 \subseteq \cdots$ is monotone and bounded by $V$, it converges in at most $|V|$ iterations to the least fixpoint $Win^r(A,F):=W^r_\infty$. By construction, $Win^r(A,F)$ is precisely the set of states from which Player 0 can force a visit to $F$. Hence, Player~0 wins the reachability game $G$ if and only if $v_0 \in Win^r(A,F)$.


Reachability games can be solved in time linear in the size $|G|$ of the game \cite{cook1970path}, where the dominant term is the transition relation, whose size is $O(|V|\cdot |\Sigma_0|\cdot |\Sigma_1|)$. Note, however, that recent tools, cf.~\cite{duret2025engineering}, use a \emph{semi-symbolic} representation of arenas, where the states are explicit, while the alphabets are represented symbolically, which is aimed at addressing the ``alphabet-explosion problem.''

In a \emph{B\"uchi game} $G=\langle A,B\rangle$ the winning condition $\alpha$ is defined by a subset $B\subseteq V$ of B\"uchi states. A play is winning for Player~0 if states from $B$ occur \emph{infinitely often} in $\pi$. Equivalently, Player~0 wins if it can force infinitely many visits to $B$. Since $v_0$ is non-recurrent, $v_0\not\in B$.
 
To solve a B\"uchi game $G=\langle A,B\rangle$, we first compute the largest subset of $B$ from which Player~0 can enforce infinitely many returns to the $B$ region \cite{vardi1984automata}. 
Define $ B_0 := B$, and, for every $i\ge 0$, 
$B_{i+1} := \{v \in B_i: v_0\in Win^r(A[v_0 \mapsto  v],B_i)\}$.
Note that $B_{i+1}$ can be computed by solving the reachability game
$G=\langle A,Force_0(B_i)\rangle$ and intersecting the winning states with $B_i$.
The sequence $B_0 \supseteq B_1 \supseteq B_2 \supseteq \cdots $ is monotone decreasing and therefore converges in at most $|V|$ iterations to a greatest fixpoint $B_\infty$. The set $B_\infty$ is the largest subset of $B$ such that every state in $B_\infty$ can force a return to $B_\infty$, and hence can participate in infinitely many visits to $B$. We then define $Win^B(A,B) := Win^r(A,B_\infty).$ Thus, Player~0 wins the B\"uchi game $G$ if and only if $v_0 \in Win^B(A,B).$

Finally, in a \emph{reachability-B\"uchi} game $G=\langle A,F,B \rangle$ the winning condition is defined by two subsets $F,B\subseteq V$. A play $\pi$ is in $\alpha$ if a state from $F$ occurs at least once in $\pi$ or a state from $B$ occurs infinitely often in $\pi$. We can assume, wlog, that $F\cap B=\emptyset$, since a visit to $F$ suffices for Player 0 to win the game.

A reachability--B\"uchi game combines a reachability objective with a B\"uchi
objective. To solve such a game, we compute the largest subset of the B\"uchi goal from which Player~0 can enforce another visit either to the B\"uchi objective or to the reachability objective. Define $B_0 := B$, and, for every $i\ge 0$, 
$B_{i+1} := \{v \in B_i: v_0\in Win^r(A[v_0 \mapsto  v],F\cup B_i)\}$.
Again, the resulting descending chain converges in at most $|V|$ iterations to a greatest fixpoint $B_\infty$. The set $B_\infty$ is the largest subset of $B$ such that every state in $B_\infty$ can force another visit to $F\cup B_\infty$. Thus, from every state in $B_\infty$, Player~0 can either reach the target set $F$ or return to the recurrent B\"uchi region.
We then define $Win^{rB}(A,F,B) := Win^r(A,F\cup B_\infty)$. Thus, Player~0 wins the reachability--B\"uchi game $G=\langle A,F,B\rangle$ if and only if $v_0 \in Win^{rB}(A,F,B)$. Since each reachability computation requires time $O(\|G\|)$ and the greatest fixpoint converges in at most $|V|$ iterations, the overall running time is $O(|V|\cdot\|G\|)$.

\begin{lemma}
Player 0 wins the reachability-B\"uchi game $G=\langle A,F,B \rangle$ iff $v_0\in Win^{rB}(A,F,B)$.
\end{lemma}

\begin{proof}
Assume first that $v_0\in Win^{rB}(A,F,B)=Win^r(A,F\cup B_\infty)$. Then Player~0 has a strategy that forces a visit to $F\cup B_\infty$. If a play reaches $F$, then Player 0 wins immediately. Otherwise, the play reaches some $b\in B_\infty$. 
Since $B_\infty$ is a fixpoint of the map $X\mapsto \{\,v\in X : v_0\in Win^r(A[v_0\mapsto v],F\cup X)\,\}$, we have that for every $b\in B_\infty$, $v_0\in Win^r(A[v_0\mapsto b],F\cup B_\infty)$.
Hence, from a visit to $b\in B_\infty$, Player~0 can force a visit to
$F\cup B_\infty$. 
Repeating this argument yields either an eventual visit to $F$ or infinitely many visits to $B_\infty\subseteq B$. In both cases the reachability–Büchi objective is satisfied, so Player~0 has a winning strategy from $v_0$. 

\vspace{1ex}

Assume now that $v_0\notin Win^{rB}(A,F,B)=Win^r(A,F\cup B_\infty)$.
By correctness of reachability games, Player~1 has a strategy that avoids
$F\cup B_\infty$ forever. Consider the descending chain $B_0\supseteq B_1\supseteq \cdots$. For every state $b\in B\setminus B_\infty$, let $k$ be the least index such that $b\in B_k\setminus B_{k+1}$, and call $k$ the rank of $b$. Since $b \in B_k$ and $b \notin B_{k+1}$, and
$B_{k+1} = \{v \in B_k : v_0 \in \mathrm{Win}^r(A[v_0 \mapsto v], F \cup B_k)\}$,
it follows that $v_0 \notin \mathrm{Win}^r(A[v_0 \mapsto b], F \cup B_k)$. Hence, whenever the play visits $b$, Player~1 can prevent any future
visit to $F\cup B_k$.
Consequently, every later visit to \(B\) must belong to \(B\setminus B_k\), and therefore has strictly smaller rank. 
Since ranks are natural numbers, they can decrease only finitely many times. Thus every play consistent with Player~1's strategy visits \(B\) only finitely often. Moreover, \(F\) is never reached. Therefore, the reachability--B\"uchi objective is not satisfied, and Player~0 has no winning strategy from \(v_0\). 
\end{proof}

Finally, we add the notion of \emph{unsafe} states, $U\subseteq V$, where Player 0 loses immediately upon entering an unsafe state. Note that we can assume wlog that $F\cap U=\emptyset$ and $U\cap B=\emptyset$. To account for unsafe states, we simply have to adapt the forcing function, with the new definition: $Force_0(W, U)$ is the set of states $w\not\in U$ where there exists an action $\sigma^0\in\Sigma^0$ such that for all actions $\sigma^1\in\Sigma^1$ we have $\delta(w,\sigma^0,\sigma^1)\in W \setminus U$. Formally,
\[
Force_0(W,U)
=
\left\{
w\in V\setminus U
\;\middle|\;
\exists \sigma^0\in\Sigma^0\;
\forall \sigma^1\in\Sigma^1,\;
\delta(w,\sigma^0,\sigma^1)\in W\setminus U
\right\}.
\]

Consequently, every state computed during the reachability and reachability–Büchi fixpoint constructions is guaranteed to avoid unsafe states. Using the modified forcing operator, we compute reachability while avoiding unsafe states by defining $W^r_0:=F$ and $W^r_{i+1}:=W^r_i\cup Force_0(W^r_i,U)$. As before, the resulting sequence converges to a least fixpoint, which we denote by $Win^r(A,U,F):=W^r_\infty$.

We provide below pseudocode descriptions of the algorithms introduced above.

\begin{algorithm}[!h]
\caption{\textsc{SolveReachability}$(A,U,F)$}
\begin{algorithmic}[1]
\State $W \gets F$
\Repeat
    \State $W_{\mathit{old}} \gets W$
    \State $W \gets W \cup \textsc{Force$_0$}(A,U,W)$
\Until{$W = W_{\mathit{old}}$}
\State \Return $W$
\end{algorithmic}
\end{algorithm}


\begin{algorithm}[!h]
\caption{\textsc{SolveReachabilityBuchi}$(A,U,F,B)$}
\begin{algorithmic}[1]
\State $B_0 \gets B$
\Repeat
    \State $B_{\mathit{old}} \gets B_0$
    \State $X \gets \textsc{SolveReachability}(A,U,\textsc{Force}_0(B_0 \cup F))$
    \State $B_0 \gets B_0 \cap X$
\Until{$B_0 = B_{\mathit{old}}$}
\State \Return \textsc{SolveReachability}$(A,U,F \cup B_0)$
\end{algorithmic}
\end{algorithm}

%% file: 2-rabin_algorithms.tex
\section{Rabin's Algorithm for Solving Rabin's Games}

In this section, we describe the algorithm implicit in \cite{rabin1972automata}. The description there is not as explicit as the description provided here for two reasons. First, Rabin's presentation combines the description of the algorithm with a proof that Rabin games enjoy the \emph{positional-determinacy} property (cf. \cite{gradel2004positional}); that is, whenever a player has a winning strategy, it also has a memoryless winning strategy, represented by a function $s:  V \rightarrow\Sigma_0$. Second, the framework in  \cite{rabin1972automata} dispenses with the alphabets $\Sigma^0$ and $\Sigma^1$ by absorbing them into the state set $V$ of the arena. This simplifies the notation, but makes it more difficult to pinpoint the impact of alphabet size on the running time of the algorithm. Since in the reactive-synthesis applications studied in \cite{rabin1972automata}, as well as in many later works, e.g., \cite{pnueli1989synthesis}, the alphabet is exponential in the size of the temporal-specification formula, we choose here to treat the alphabet explicitly rather than implicitly, in order to study its impact on the complexity of the algorithm. The explicit description of Rabin's algorithm will lead to an improved algorithm in the next section.

\subsection{Rabin Games and Extended Rabin Games}

In a Rabin game, the winning condition is specified by a sequence $\rho$ of pairs $ (L_i,R_i)$ of subsets of $V$; each pair is called a \emph{Rabin pair}. A play satisfies a Rabin pair $( L, R)$ if $L$ is visited infinitely often and $R$ is visited only finitely often along the play. We can assume wlog that $L\cap R=\emptyset$. A play satisfies the Rabin condition $\rho$ if it satisfies at least one Rabin pair in $\rho$. The \emph{Rabin index} of the game is the length of $\rho$.

Here we also consider \emph{extended} Rabin games, whose winning condition is specified by a tuple $\langle U, F,B,\rho\rangle$, which combines unsafe states $U$, a reachability goal $F$, a B\"uchi goal $B$, and a Rabin condition $\rho$. A play $\pi$ satisfies this winning condition if it satisfies the reachability goal, the B\"uchi goal, \emph{or} the Rabin goal, while avoiding unsafe states. Without loss of generality, we can assume that for each Rabin pair $(L,R)$ we have $(L\cup R)\cap U=\emptyset$. We can also  assume that no Rabin pair $( L,R)$ has $R=\emptyset$, since in this case we can just drop this pair and replace $B$ by $B\cup L$. 
Recall also that $v_0$ cannot occur in the winning condition $\langle U, F,B,\rho\rangle$.

\subsection{Reduction to Reachability-Rabin Games}
We now show that solving extended Rabin games can be reduced to solving a special class of extended Rabin games with a trivial B\"uchi condition. A \emph{reachability-Rabin game} is an extended Rabin game whose winning
condition is of the form $\langle U, F,\emptyset,\rho\rangle$ where $U$ is a set of unsafe states, $F$ is a reachability goal, $\emptyset$ is the trivial B\"uchi condition, and $\rho$ is a Rabin
condition. For a reachability--Rabin game $G=\langle A,\alpha \rangle$, we write $Win^{rR}(A,\alpha)$ for the set of states from which Player~0 has a winning strategy,
that is, the set of states $v$ such that Player 0 has a winning strategy in the game $G_v=(A_v,\alpha)$, where $A_v=A[v_0\mapsto v]$.
Suppose that this set can be computed for reachability--Rabin games. We show how this can be used to solve extended Rabin games.

Consider an extended Rabin game $G=\langle A,\alpha \rangle $ with $\alpha=\langle U, F,B,\rho\rangle$. Define $B_0=B$, and, for every $i\ge 0$, 
$B_{i+1}=\{v\in B_i: v_0\in Win^{rR}(A_v,\alpha_i)\}$
where $A_v=A[v_0\mapsto v]$ and
$\alpha_i=\langle U,F\cup B_i,\emptyset,\rho\rangle$.
Intuitively, Player 0 is now aiming to force another visit to $B_i$ or $F$, or otherwise win by satisfying the Rabin condition $\rho$. 
As before, this descending sequence converges in at most
\(|V|\) iterations to a greatest fixpoint \(B_\infty\).

\begin{lemma}
Player 0 wins the extended Rabin game $G=\langle A,\alpha \rangle $, with $\alpha=\langle U,F,B,\rho \rangle$ iff $v_0\in  Win^{rR}(A,\alpha_\infty)$, where $\alpha_\infty=\langle U,F\cup B_\infty,\emptyset,\rho\rangle$.
\end{lemma}

\begin{proof}
\textbf{($\Leftarrow$)}
Assume $v_0 \in Win^{rR}(A,\alpha_\infty)$ and let $\sigma$ be a winning strategy for $(A,\alpha_\infty)$ from $v_0$, where $\alpha_\infty=\langle U,F\cup B_\infty,\emptyset,\rho\rangle$.

Let $\pi$ be a play consistent with $\sigma$.  Since $\sigma$ is winning for $\alpha_\infty$, the play $\pi$
avoids $U$ and satisfies one of the following: $\pi$ reaches $F\cup B_\infty$ at least once, or
$\pi$ satisfies the Rabin condition $\rho$.
If $\pi$ satisfies $\rho$, then it satisfies the extended objective $\alpha=\langle U,F,B,\rho\rangle$.
If $\pi$ reaches $F$, then it also satisfies $\alpha$.
Assume therefore that $\pi$ does not satisfy $\rho$ and does not reach $F$. Then $\pi$ must reach  a state $b \in B_\infty$.
By definition of $B_\infty$, we have $B_\infty\subseteq Win^{rR}(A,\alpha_\infty)$. Thus, Player~0 has a winning strategy for $(A,\alpha_\infty)$ from $b$. From this point on, Player~0 switches to such a strategy. 
Applying the same argument repeatedly to the suffix of $\pi$ starting at $b$, either the suffix satisfies $\rho$, or reaches $F$, or reaches $B_\infty$ again.

Since by assumption $\pi$ never satisfies $\rho$ and never reaches $F$, the suffix must reach $B_\infty$ again. Repeating this argument inductively shows that $\pi$ visits $B_\infty$ infinitely often. As $B_\infty \subseteq B$, it follows that $\pi$ visits $B$ infinitely often. Hence, $\pi$ satisfies the extended objective $\alpha$. Since $\pi$ was arbitrary, $\sigma$ is winning for the extended Rabin game from $v_0$.

\textbf{($\Rightarrow$)} We prove the contrapositive. Assume $v_0\notin Win^{rR}(A,\alpha_\infty)$.
Fix a Player-$1$ strategy $\tau_\infty$ that spoils $\alpha_\infty$ from $v_0$. Let $\pi$ be a play consistent with $\tau_\infty$. If $\pi$ reaches $U$, then $\pi$ is losing for $\alpha$. Otherwise, $\pi$ avoids $U$ and, since $\tau_\infty$ spoils $\alpha_\infty$, it does not reach $F\cup B_\infty$
and does not satisfy $\rho$.

It remains to show that Player~$1$ can ensure that $B$ is visited only finitely often. For $b\in B\setminus B_\infty$ define $\mathrm{rank}(b):=\min\{\,k\ge 0 \mid b\in B_k\setminus B_{k+1}\,\}$. Suppose $b\in B_k\setminus B_{k+1}$, and let $\alpha_k=\langle U,F\cup B_k,\emptyset,\rho\rangle$. Since $b\not\in B_{k+1}$, we must have $b\notin Win^{rR}(A,\alpha_k)$, hence Player~$1$ has a strategy
$\tau_k^b$ that spoils $\alpha_k$ from $b$. 
Let $\pi$ be a $U$-avoiding play consistent with $\tau_\infty$, and suppose that $\pi$ does visit $B \setminus B_\infty$; let $b$ be the first visited state in $B\setminus B_\infty$ with $k=\mathrm{rank}(b)$. Define a Player~$1$ strategy $\tau$ that follows $\tau_\infty$ until the first visit to some $b\in B\setminus B_\infty$, and from then on follows $\tau_k^b$, where $k=\mathrm{rank}(b)$. Since $\tau_k^b$ spoils $\alpha_k$, the suffix from $b$ avoids $F\cup B_k$ and does not satisfy $\rho$. In particular, it avoids $B_k$, hence no state of $B$ with rank $\ge k$ can appear again. Therefore, a later visit to $B$ must have strictly smaller rank. This argument can be repeated to show that Player 1 has a strategy where $B$ can be visited only finitely many times, such that $F$ is never reached and $\rho$ is spoiled, so $\alpha$ is spoiled. Hence Player~$0$ does not win $G$ from $v_0$.

This proves the contrapositive and completes the proof. 
\end{proof}

\subsection{Solving Reachability-Rabin Games by induction on Rabin number}

Consider now a reachability-Rabin game $G=\langle A,\alpha \rangle$, with $\alpha=\langle U,F,\emptyset,\rho \rangle$,  where $\rho=\{(L_1,R_1),\ldots$, $ (L_k,R_k)\}$. Let $R=\bigcup_{i=1}^k R_i$. We can assume that $R\cap F=\emptyset$, since a play that hits $F$ is already winning for Player 0. We call the states in $R$ \emph{Rabin states}, and call $|R|$ the \emph{Rabin number} of $G$.

We first solve the reachability game $G_r= \langle A,U,F \rangle$, and compute $Win^r(G_r)$. Clearly, $Win^r(G_r) \subseteq Win^{rR}(G)$. Thus, if $v_0\in Win^r(G_r)$, then $v_0 \in Win^{rR}(G)$.
Moreover, if a play enters a state $u\in Win^r(G_r)$, then Player 0 can continue from there with the strategy that placed $u$ in $Win^r(G_r)$. Hence, all states in $Win^r(G_r)$ can be declared to be \emph{reachability states} by using the revised winning condition $\alpha_r=(U,F\cup Win^r(G_r),\emptyset, \rho-Win^r(G_r))$, where $\rho-Win^r(G_r)$ is obtained from $\rho$ by replacing each Rabin pair $(L,R)$ by the pair $(L-Win^r(G_r),R-Win^r(G_r))$. Thus, from now on, we assume that when we have a reachability-Rabin game $G=\langle A,\alpha \rangle $, with $\alpha=\langle U,F,\emptyset,\rho \rangle$, we have that $Win^r(G) \subseteq F$. If $v_0\in Win^r(G)$, then the game is \emph{trivial}, which we assume is not the case.



We now develop an algorithm for solving reachability-Rabin games by induction on the Rabin number.  Note that when the Rabin number is 0, all Rabin pairs are of the form $(L_i,\emptyset)$. Hence, the reachability-Rabin condition $\langle U, F,\emptyset, \rho\rangle$ is equivalent to the reachability-B\"uchi condition $\langle U, F,\cup_{i=1}^k L_i\rangle$. Thus, we are back to reachability-B\"uchi games with unsafe states, solved in the previous section. We now show that solving Rabin games with Rabin number $n$ reduces to solving Rabin games with Rabin number $n-1$. 

Consider a Rabin condition $\rho=\{(L_1,R_1),\ldots,(L_k,R_k)\}$. For $v\in V$, let $\rho-\{v\}$ be obtained by replacing each Rabin pair $(L_i,R_i)$ by $(L_i-\{v\},R_i-\{v\})$. A Rabin pair $(\emptyset,R)$ can be discarded, since it cannot be satisfied by any play. Rabin pairs with empty right-hand side can be converted to a B\"uchi condition, as discussed above.

A \emph{game tree} is a mapping $g: (\Sigma^1)^* \rightarrow V$.
Suppose Player 0 has a winning strategy $s: (\Sigma^1)^* \rightarrow \Sigma^0$. The \emph{game tree} $g_s: (\Sigma^1)^* \rightarrow V$ for $s$ is defined inductively: $g_s(\varepsilon)=v_0$, and, for $x\in (\Sigma^1)^*$ and $\sigma^1\in \Sigma^1$, we have $g_s(x\cdot \sigma^1)=\delta(g_s(x),s(x),\sigma^1)$. We say that $g_s$ is a \emph{winning game tree} if $s$ is a winning strategy. A game tree $g': (\Sigma^1)^* \rightarrow V$ is a \emph{subtree} of a game tree $g: (\Sigma^1)^* \rightarrow V$ \emph{rooted} at $x\in (\Sigma^1)^*$ if $g'(y)=g(xy)$ for all $y\in (\Sigma^1)^*$.

Let $G=\langle A,\alpha\rangle$ be a reachability-Rabin game with arena $A=\langle V,v_0, \Sigma^0,\Sigma^1, \delta \rangle$ and winning condition $\alpha=\langle U, F,\emptyset, \rho\rangle$. Assume that Player 0 has a winning strategy $s: (\Sigma^1)^* \rightarrow \Sigma^0$ in $G$, with game tree $g_s$. 
When the game hits a state in $F$, Player 0 can declare victory and stop the play. Thus, wlog, we can delete from $g_s$ all nodes $xy$, where $g_s(x)\in F$ and $y\not=\varepsilon$. Equivalently, we can assume that in $g_s$, all nodes labeled by states in $F$ are leaves. We say that $g_s$ is \emph{pruned}. 
Still, $g_s$ must have infinitely many nodes, since otherwise Player 0 has a winning strategy for $G_r$ from $v_0$, which would make the game trivial.

\begin{lemma} [The State-Avoidance Lemma]
Let $G=\langle A,\alpha\rangle$ be a reachability-Rabin game
with arena $A=\langle V,v_0, \Sigma^0,\Sigma^1, \delta \rangle$, winning condition $\alpha=\langle U, F,\emptyset, \rho\rangle$, and Rabin number $n>0$. We assume wlog that $Win^r(A,U,F)\subseteq F$. Suppose that Player 0 has a winning strategy in $G$ from a state $v\not\in F$. Let  $s: (\Sigma^1)^* \rightarrow \Sigma^0$ be a winning strategy for Player 0 in $G$ from $v$, with pruned game tree $g_s$. Then there exists a state $v'\in R$ such that $g_s$ has an infinite subtree $g'$ in which $v'$ does not occur.
\end{lemma}
\begin{proof}
Let $R=\{v_1,\ldots,v_n\}$. If $g_s$ has an infinite subtree $g'$ where $v_1$ does not occur in $g'$, we are done. Otherwise, $v_1$ must occur in $g_s$, that is, there is some $x_1\in (\Sigma^1)^*$ such that $g_s(x_1)=v_1$. The subtree $g_{x_1}$ rooted at $x_1$ is a winning game tree for Player 0 starting from $v_1$. If $g_{x_1}$ is finite, then Player 0 wins the reachability game with unsafe set $U$ and
target $F$ from $v_1$, so $v_1\in Win^r(A,U,F)$. By assumption, $v_1\in F$. Hence $v_1\not\in R$ and therefore $g_{x_1}$ is not finite.


So, assume that $g_{x_1}$ is infinite. We can now repeat the same argument. If $g_{x_1}$ has an infinite subtree $g'$ where $v_2$ does not occur in $g'$, we are done. Otherwise, $v_2$ must occur in $g_{x_1}$, that is, there is some $x_2\in (\Sigma^1)^*$ such that $g_{x_1}(x_2)=v_2$. The subtree $g_{x_2}$ rooted at $x_2$ is a winning game tree for Player 0 starting from $v_2$. Again, if $g_{x_2}$ is finite then it must be that Player 0 wins the   reachability game with unsafe states $U$ and target $F$ from $v_2$, so $v_2\in Win^r(A,U,F)$, which implies that $v_2\in F$, and $v_2\not\in R$. It follows that $g_{x_2}$  is also infinite.

By repeating this argument, going through $v_1,\ldots,v_n$ in a cycle, we have found an infinite branch $p$ of $g_s$ on which each of $v_1,\ldots,v_n$ is visited infinitely often. Thus, the Rabin condition $\rho$ is not satisfied, and so $p$ must visit $F$. But $g_s$ is pruned whenever $F$ is visited, so the pruned tree $g_s$ cannot contain an infinite branch that visits $F$.

It follows that there exists a state $v'\in R$ such that $g_s$ has an infinite subtree $g'$ in which $v'$ does not occur in $g'$. 
\end{proof}

Let $G=\langle A,\alpha\rangle$ be a reachability--Rabin game with arena $A=\langle V,v_0, \Sigma^0,\Sigma^1, \delta \rangle$, winning condition $\alpha=\langle U, F,\emptyset, \rho\rangle$, and Rabin number $n>0$, where $Win^r(A,U,F)\subseteq F$ and $v_0\not\in F$. Let  $s: (\Sigma^1)^* \rightarrow \Sigma^0$ be a winning strategy for Player 0 in $G$, with pruned game tree $g_s$. By the State-Avoidance Lemma, there exists a state $v'\in R$ such that $g_s$ has an infinite subtree $g'$ rooted at some node $x\in(\Sigma^1)^*$, where $v'$ does not occur in $g'$. Let $g_s(x)=v$. Then Player 0 has a winning strategy in the game $G=\langle A[v_0\mapsto v],\alpha\rangle$ such that no play visits the state $v'$. But then we might as well consider $v'$ to be an unsafe state, since the play never enters it. This gives the following reduction:

\begin{lemma}
Let $G=\langle A,\alpha\rangle$ be a nontrivial reachability--Rabin game with arena $A=\langle V,v_0, \Sigma^0,\Sigma^1, \delta \rangle$, and winning condition $\alpha=\langle U, F,\emptyset, \rho\rangle$, where $Win^r(A,U,F)\subseteq F$. If Player 0 has a winning strategy $s$ from a state $v$ such that some state $v'\in R$ does not occur in $g_s$, then Player 0 has a winning strategy from $v$ in the extended Rabin game  $G'=\langle A,\alpha'\rangle$, with the winning  condition $\alpha'=\langle U\cup \{v'\}, F-\{v'\},\emptyset, \rho-\{v'\}\rangle$.
\end{lemma}
Note that if $G$ has Rabin number $n$, then $G'$ has Rabin number smaller than $n$. As we saw earlier, when we delete states from a Rabin condition, we may obtain Rabin pairs with empty right-hand sides, which then need to be consolidated into a B\"uchi condition.
Combining the lemmas, we obtain:

\begin{theorem}
Let $G=\langle A,\alpha\rangle$ be a reachability-Rabin game with arena $A=\langle V,v_0, \Sigma^0,\Sigma^1, \delta \rangle$ and winning condition $\alpha=\langle U, F,\emptyset, \rho\rangle$, where $Win^r(A,U,F)\subseteq F$. If $G$ is not trivial, then Player 0 has a winning strategy in $G$ iff there are states $v$ and $v'\in R$ such that:
\begin{enumerate}
\item 
Player 0 has a winning strategy in the extended Rabin game $G'=\langle A[v_0\mapsto v],\alpha'\rangle$, where $\alpha'=\langle U\cup \{v'\}, F-\{v'\},\emptyset, \rho-\{v'\}\rangle$, and
\item 
Player 0 has a winning strategy in the extended Rabin game  $G''=\langle A,\alpha''\rangle$, where $\alpha''=\langle U, F\cup \{v\},\emptyset, \rho-\{v\}\rangle$.
\end{enumerate}
\end{theorem}
We call $G'$ the game obtained from $G$ by \emph{declaring} $v'$ \emph{unsafe}.
In practice, there  is no need to enumerate all pairs $(v,v')$ of states.
\begin{corollary}
Let $G=\langle A,\alpha\rangle$ be a reachability-Rabin game with arena $A=\langle V,v_0, \Sigma^0,\Sigma^1, \delta \rangle$ and winning condition $\alpha=\langle U, F,\emptyset, \rho\rangle$, and Rabin states $R=\{v_1,\ldots,v_n\}$, where $Win^r(G)\subseteq F$. If $G$ is not trivial, then Player 0 has a winning strategy in $G$ iff Player 0 has a winning strategy in the game $G_W$ obtained from $G$ by declaring all states in $W$ to be reachability states, where $W=\bigcup_{i=1}^k Win^{rR} (G_i) \not=\emptyset$, and $G_i$ is the extended Rabin game obtained by declaring $v_i$ unsafe.
\end{corollary}
Note that all the $G_i$ games, as well as $G_W$, have a Rabin number smaller than $n$.

\subsection{Pseudocode and Complexity Analysis}
We now give pseudocode for the Rabin-number-based algorithm. We first introduce three auxiliary procedures for deleting states and normalizing Rabin conditions.

\captionof{algorithm}{Auxiliary procedures}
\label{alg:auxiliary}
\label{alg:auxiliary}
\noindent\rule{\linewidth}{0.4pt}
\begin{algorithmic}[1]

\Function{\textsc{DeleteState}}{$\rho,v$}
    \State $\rho' \gets \emptyset$
    \ForAll{$(L,R)\in\rho$}
        \State $L' \gets L-\{v\}$; $R' \gets R-\{v\}$
        \If{$L'\neq\emptyset$}
            \State append $(L',R')$ to $\rho'$
        \EndIf
    \EndFor
    \State \Return $\rho'$
\EndFunction

\Function{\textsc{DeleteStates}}{$\rho,X$}
    \State $\rho' \gets \emptyset$
    \ForAll{$(L,R)\in\rho$}
        \State $L' \gets L-X$; $R' \gets R-X$
        \If{$L'\neq\emptyset$}
            \State append $(L',R')$ to $\rho'$
        \EndIf
    \EndFor
    \State \Return $\rho'$
\EndFunction

\Function{\textsc{NormalizePairs}}{$\rho$}
    \State $B \gets \emptyset$; $\rho' \gets \emptyset$
    \ForAll{$(L,R)\in\rho$}
        \If{$L=\emptyset$}
            \State \textbf{continue}
        \ElsIf{$R=\emptyset$}
            \State $B\gets B\cup L$
        \Else
            \State append $(L,R)$ to $\rho'$
        \EndIf
    \EndFor
    \State \Return $(B,\rho')$
\EndFunction

\end{algorithmic}
\noindent\rule{\linewidth}{0.4pt}

\vspace{0.2cm}
We say that a Rabin condition $\rho$ is \emph{normalized} if every
pair $(L,R)\in\rho$ satisfies $L\neq\emptyset$ and $R\neq\emptyset$. The procedure \textsc{NormalizePairs} in Algorithm~3 computes a normalized condition: pairs with $L=\emptyset$ are discarded, since they cannot be satisfied, while pairs with $R=\emptyset$ are absorbed into the B\"uchi set $B$. Normalization provides a simple base case when the resulting $\rho$ is empty and a canonical representation for memoization. In particular, recursive calls that yield the same normalized subproblem can be identified and solved only once, as used in the complexity analysis below.

\begin{algorithm}[!h]
\caption{\textsc{SolveExtendedRabin}$(A,U,F,B,\rho)$}
\begin{algorithmic}[1]
\State $W_r \gets \textsc{SolveReachability}(A,U,F)$
\State $F \gets F \cup W_r$; $B \gets B - W_r$
\State $\rho \gets \textsc{DeleteStates}(\rho,W_r)$
\If{$v_0 \in F$}
    \State \Return $F$
\EndIf
\Repeat
 \State $B_{\mathit{old}} \gets B$
  \State $X \gets \textsc{SolveReachabilityRabin}(A,U,F \cup B,\rho)$
    \State $B \gets B \cap X$
\Until{$B = B_{\mathit{old}}$}
\State \Return $\textsc{SolveReachabilityRabin}(A,U,F \cup B,\rho)$
\end{algorithmic}
\end{algorithm}

\begin{remark}
Lines 4-6 of Algorithm 4 are an early-termination optimization. If $v_0 \in F$, then Player~0 is already known to win from the initial state, and the remaining fixpoint computation is therefore unnecessary. In this case, however, the returned set $F$ need not coincide with the full winning region of the extended Rabin game, since additional states may be winning through the B\"uchi or Rabin objectives.
\end{remark}

\paragraph{Complexity Analysis:}
Let $n:=|V|$, let $M:=|G|$ be the size of the game graph, let
$k:=|\rho|$ be the Rabin index, and let
$r := \left|\bigcup_{(L_i,R_i)\in \rho} R_i\right|$
be the Rabin number. A naive analysis of the size of the recursion tree, i.e., all possible orders of eliminating states or pairs, would yield factorial bounds with respect to $n$ (for number-based recursion). With memoization, in contrast, each normalized subproblem is solved only once. Thus, the complexity is governed by the number of distinct subproblems.
\noindent
In Algorithm 5, Rabin-reachability subproblems are generated by changing the sets $U,F \subseteq V$ and by deleting Rabin states from $\rho$. Thus, the number of distinct memoized subproblems is bounded by $2^{O(n)}$. Since each subproblem performs only polynomial-time fixpoint computations and polynomial-time normalization, the memoized complexity is
$T_{\mathrm{num}}^{\mathrm{memo}}
=
2^{O(n)}\cdot poly(n,M,k)$.
Here the Rabin index $k$ occurs only in the term $poly(n,M,k)$. Since the Rabin index $k$ can be exponentially larger than the Rabin number $r$, this analysis suppresses the complexity in $k$, which we come back to below. Note that the size of the alphabet is a linear factor in the size $M$ of the game, while absorbing the alphabet into the state set as in \cite{rabin_finite_1959} would yield complexity that is exponential in the size of the alphabet.


\begin{algorithm}[!h]
\caption{\textsc{SolveReachabilityRabin-Number}$(A,U,F,\rho)$}
\begin{algorithmic}[1]
\State $W_r \gets \textsc{SolveReachability}(A,U,F)$
\State $F \gets F \cup W_r$
\State $\rho \gets \textsc{DeleteStates}(\rho,W_r)$
\State $(B_0,\rho) \gets \textsc{NormalizePairs}(\rho)$
    \Comment{move pairs $(L,\emptyset)$ into the B\"uchi set}
\If{$|\rho|=0$}
     \State \Return $\textsc{SolveReachabilityBuchi}(A,U,F,B_0)$
\EndIf
\State $R \gets \textsc{RightUnion}(\rho)$ \Comment{$\bigcup\{R \mid (L,R) \in \rho\}$, the set of all Rabin states}
\State $W \gets \emptyset$

\ForAll{$v \in R$}
\State $W_v \gets \textsc{SolveExtendedRabin}(A,\;U \cup \{v\},\;F - \{v\},\;B_0,\;\textsc{DeleteState}(\rho,v))$
    \State $W \gets W \cup W_v$
\EndFor
\State \Return $\textsc{SolveExtendedRabin}(A,\;U,\;F \cup W,\;B_0,\;
\textsc{DeleteStates}(\rho,W))$
\end{algorithmic}
\end{algorithm}

\begin{remark}
Hausmann et al. \cite{hausmann2024symbolic} describe an approach to solve Emerson-Lei games, which generalize Rabin games. Rabin's algorithm \cite{rabin1972automata} can be viewed as a special case of their algorithm. While Rabin's algorithm is expressed here as a recursive algorithm, the Hausmann et al. algorithm first constructs a \emph{Zielonka tree} \cite{zielonka1998infinite}, and uses it to derive a $\mu$-calculus expression over the game arena. The Zielonka tree, however, may be of exponential size. Expressing the algorithm as a recursive algorithm as done here facilitates memoization, which, as we show above, reduces the complexity.
\end{remark}

%% file: 3-new_algorithm_induction.tex
\section{Solving Reachability-Rabin Games by induction on Rabin index}

This section provides a novel algorithm inspired by Rabin's algorithm. In Rabin's algorithm we recursively remove one state at a time to simplify the underlying algorithmic problem. The key observation of this section is that, if the removed sets are chosen carefully, one can remove an entire set of states at once. The potential advantage of this reduction is in the recursion depth. While the Rabin-number-based algorithm eliminates one Rabin state at a time, the index-based reduction eliminates an entire right-hand side $R_j$ and thereby reduces the Rabin index by at least one. Thus, when the Rabin index $k$ is substantially smaller than the Rabin number $r$, a recursive branch may require substantially fewer reduction steps. 

Consider again the reachability--Rabin game $G=\langle A,\alpha \rangle$, with $\alpha=\langle U,F,\emptyset,\rho \rangle$,  where $\rho=\{(L_1,R_1),$ $\ldots,(L_k,R_k)\}$. Let $R=\bigcup_{i=1}^k R_i$. As before, we assume that $R\cap F=\emptyset$, since a play that reaches $F$ is already winning for Player 0. Recall that the states in $R$ are Rabin states, and $|R|$ is the \emph{Rabin number} of $G$.

Consider the reachability game $G_r=\langle A,U,F\rangle$ and compute $Win^r(G_r)$. Clearly, $Win^r(G_r) \subseteq Win^{rR}(G)$. Thus, if $v_0\in Win^r(G_r)$, then $v_0 \in Win^{rR}(G)$.
Furthermore, if a  play enters a state $u\in Win^r(G_r)$, then  Player 0 can continue from there with the strategy that placed $u$ in $Win^r(G_r)$. We declare all states in $Win^r(G_r)$ to be \emph{reachability states} by using the revised winning condition $\alpha_r=(U,F\cup Win^r(G_r),\emptyset, \rho-Win^r(G_r))$, where $\rho-Win^r(G_r)$ is obtained from $\rho$ by replacing each Rabin pair $(L,R)$ by the pair $(L-Win^r(G_r),R-Win^r(G_r))$. Thus, as before, we  assume that when we have a reachability--Rabin game $G=\langle A,\alpha \rangle$, with $\alpha=\langle U,F,\emptyset,\rho \rangle$, we have $Win^r(G) \subseteq F$ and $v_0 \not \in F$.

Consider a Rabin condition $\rho=\{(L_1,R_1),\ldots,(L_k,R_k)\}$. For a right-hand side $R_j$, let $\rho-R_j$ be obtained from $\rho$ by replacing every Rabin pair $(L_i,R_i)$ with $(L_i-R_j,R_i-R_j)$. A Rabin pair $(\emptyset,R)$ can be discarded, since it can play no role in winning a play. Rabin pairs with empty right-hand side can be converted to a B\"uchi condition, as discussed above. Since the pair $(L_j,R_j)$ becomes $(L_j-R_j,\emptyset)$ in $\rho-R_j$, we reduce the Rabin index by going from $\rho$ to $\rho-R_j$.

\begin{lemma} [The Set-Reduction Lemma]\label{index_set_avoidance_lemma}
Let $G=\langle A,\alpha\rangle$ be a reachability-Rabin game
with arena $A=\langle V,v_0, \Sigma^0,\Sigma^1, \delta \rangle$, winning condition $\alpha=\langle U, F,\emptyset, \rho\rangle$, with $\rho=\{(L_1,R_1),\ldots,(L_k,R_k)\}$. We assume wlog that $Win^r(G)\subseteq F$. Suppose that Player 0 has a winning strategy in $G$ from a state $v\not\in F$. Let  $s: (\Sigma^1)^* \rightarrow \Sigma^0$ be a winning strategy for Player 0 in $G$ from $v$, with pruned game tree $g_s$. Then there exists some $1\leq j \leq k$ such that $g_s$ has an infinite subtree $g'$ in which no state of $R_j$ occurs.
\end{lemma}
\begin{proof}
If $g_s$ has an infinite subtree $g'$ in which no state of $R_1$ occurs, we are done. Otherwise, some $v_1\in R_1$ must be in the range of $g_s$, that is, there is some $x_1\in (\Sigma^1)^*$ such that $g_s(x_1)=v_1$. The subtree $g_{x_1}$ rooted at $x_1$ is a winning game tree for Player 0 starting from $v_1$. If $g_{x_1}$ is finite, then Player 0 wins the reachability game with unsafe states $U$ and target $F$ from $v_1$, so $v_1\in Win^r(G)$. But then, by assumption, $v_1\in F$, which means that $v_1\not\in R_1$, so $g_{x_1}$ is not finite.

So assume that $g_{x_1}$ is infinite. We can now repeat the same argument. If $g_{x_1}$ has an infinite subtree $g'$  in which no state of $R_2$ occurs, we are done. Otherwise, some $v_2 \in R_2$ must be in the range of $g_{x_1}$, that is, there is some $x_2\in (\Sigma^1)^*$ such that $g_{x_1}(x_2)=v_2$. The subtree $g_{x_2}$ rooted at $x_2$ is a winning game tree for Player 0 starting from $v_2$. Again, if $g_{x_2}$ is finite then it must be that Player 0 wins the   reachability game with unsafe states $U$ and target $F$ from $v_2$, so $v_2\in Win^r(G)$, which implies that $v_2\in F$, and $v_2\not\in R_2$. It follows that $g_{x_2}$  is also infinite.

By repeating this argument, going through $R_1,\ldots,R_k$ in a cycle, we have found an infinite branch $p$ of $g_s$ where $R_1,\ldots,R_k$ are all visited infinitely often. Therefore, no Rabin pair $(L_j,R_j)$ is satisfied along $p$, since each right-hand side $R_j$ is visited infinitely often. Since $g_s$ is a
winning game tree, the branch $p$ must therefore reach $F$. But $g_s$ is pruned, so no infinite branch can continue after reaching $F$, a contradiction.

It follows that there exists some $R_j$ such that $g_s$ has an infinite subtree $g'$ in which no state of $R_j$ occurs.\end{proof}

Let $G=\langle A,\alpha\rangle$ be a reachability-Rabin game with arena $A=\langle V,v_0, \Sigma^0,\Sigma^1, \delta \rangle$, winning condition $\alpha=\langle U, F,\emptyset, \rho\rangle$, with  $\rho=\{(L_1,R_1),\ldots,(L_k,R_k)\}$, where $Win^r(G)\subseteq F$ and $v_0\not\in F$. Let  $s: (\Sigma^1)^* \rightarrow \Sigma^0$ be a winning strategy for Player 0 in $G$, with pruned game tree $g_s$. By Lemma \ref{index_set_avoidance_lemma}, there is a Rabin pair $(L_j,R_j)$ such that $g_s$ has an infinite subtree $g'$ rooted at some node $x\in(\Sigma^1)^*$,  in which no state of $R_j$ occurs. Let $g_s(x)=v$. Then Player 0 has a winning strategy in the game $G=\langle A,\alpha\rangle[v_0\mapsto v]$, such that no play visits a state in $R_j$. But then we might as well consider all states in $R_j$ to be  unsafe states, since the play never enters them. So we have obtained:
\begin{lemma}
Let $G=\langle A,\alpha\rangle$ be a nontrivial reachability-Rabin game with arena $A=\langle V,v_0, \Sigma^0,\Sigma^1, \delta \rangle$, and winning condition $\alpha=\langle U, F,\emptyset, \rho\rangle$, $\rho=\{(L_1,R_1),\ldots,(L_k,R_k)\}$, where $Win^r(G)\subseteq F$. If Player 0 has a winning strategy $s$ from a state $v$ such that, for some $1\leq j \leq k$, we have that no state of $R_j$ occurs in $g_s$, then Player 0 has a winning strategy from $v$ in the extended Rabin game  $G'=\langle A,\alpha'\rangle$, with the winning  condition $\alpha'=\langle U\cup R_j, F-R_j,\emptyset, \rho-R_j\rangle$.
\end{lemma}
Note that if $G$ has Rabin index $k$, then $G'$ has Rabin index smaller than $k$. As we saw earlier, when we delete states from a Rabin condition, we may obtain Rabin pairs with empty right-hand sides, which then need to be converted into a B\"uchi condition.

Combining the lemmas, we obtain:
\begin{theorem}
Let $G=\langle A,\alpha\rangle$ be a nontrivial  reachability-Rabin game with arena $A=\langle V,v_0, \Sigma^0,\Sigma^1, \delta \rangle$ and winning condition $\alpha=\langle U, F,\emptyset, \rho\rangle$, with $\rho=\{(L_1,R_1),\ldots,(L_k,R_k)\}$, where $Win^r(G)\subseteq F$. Then Player 0 has a winning strategy in $G$ iff there is a state $v$ and an index $1\leq j \leq k$ such that:
(1) Player 0 has a winning strategy in the extended Rabin game $G'=\langle A,\alpha'\rangle[v_0\mapsto v]$, where $\alpha'=\langle U\cup R_j, F-R_j,\emptyset, \rho-R_j\rangle$, and
(2) Player 0 has a winning strategy in the extended Rabin game  $G''=\langle A,\alpha''\rangle$, where $\alpha''=\langle U, F\cup \{v\},\emptyset, \rho-R_j\rangle$.
\end{theorem}
We call $G'$ the game obtained from $G$ by \emph{declaring} $R_j$ \emph{unsafe}.

\begin{corollary}
Let $G=\langle A,\alpha\rangle$ be a reachability-Rabin game with arena $A=\langle V,v_0, \Sigma^0,\Sigma^1, \delta \rangle$ and winning condition $\alpha=\langle U, F,\emptyset, \rho\rangle$, $\rho=\{(L_1,R_1),\ldots,(L_k,R_k)\}$, where $Win^r(G)\subseteq F$. If $G$ is not trivial, then Player 0 has a winning strategy in $G$ iff $W=\bigcup_{i=1}^k Win^{rR} (G_i) \not=\emptyset$, where $G_i$ is the extended Rabin game obtained by declaring $R_i$ unsafe, and Player 0 has a winning strategy in the game $G_W$ obtained from $G$ by declaring all states in $W$ to be reachability states. 
\end{corollary}
Note that all the $G_i$ games, as well as $G_W$, have a Rabin index smaller than $k$, or a Rabin number smaller than $n$.


\begin{algorithm}[!h]
\caption{\textsc{SolveReachabilityRabin-Index}$(A,U,F,\rho)$}
\begin{algorithmic}[1]
\State $W_r \gets \textsc{SolveReachability}(A,U,F)$
\State $F \gets F \cup W_r$
\State $\rho \gets \textsc{DeleteStates}(\rho,W_r)$
\State $(B_0,\rho) \gets \textsc{NormalizePairs}(\rho)$
\If{$|\rho|=0$}
    \State \Return $\textsc{SolveReachabilityBuchi}(A,U,F,B_0)$
\EndIf
\State $W \gets \emptyset$
\ForAll{$(L_j,R_j) \in \rho$}
    \State $W_j \gets \textsc{SolveExtendedRabin}(A,\;U \cup R_j,\;F - R_j,\;B_0,\;\textsc{DeleteStates}(\rho,R_j))$
    \State $W \gets W \cup W_j$
\EndFor
\State \Return $\textsc{SolveExtendedRabin}(A,\;U,\;F \cup W,\;B_0,\;
\textsc{DeleteStates}(\rho,W))$
\end{algorithmic}
\end{algorithm}

\paragraph{Complexity Analysis}
As in the Rabin-number-based algorithm, the subproblems generated during the execution of \textsc{SolveReachabilityRabin-Index} are indexed by a tuple $(U,F,\rho)$, with $U,F \subseteq V$, and with $\rho$ uniquely determined by the set $X \subseteq V$ of states deleted so far, since deletion acts uniformly on all pairs of $\rho$ regardless of whether $X$ is a single Rabin state, a right-hand side $R_j$, or an arbitrary winning subregion returned by a recursive call. Hence the number of distinct $\rho$'s realized during the recursion is at most $2^n$, and combined with the $2^n$ choices each for $U$ and $F$, the number of distinct memoized subproblems is bounded by $2^{O(n)}$ for both Algorithm~5 and Algorithm~6. Since each subproblem performs only polynomial-time fixpoint computations and normalization, the memoized complexity of both algorithms is $  T^{\mathit{memo}}_{\mathit{num}} \;=\; T^{\mathit{memo}}_{\mathit{index}}   \;=\; 2^{O(n)} \cdot \mathit{poly}(n,M,k).$
Memoization thus removes the factorial dependence on both the Rabin number
and the Rabin index, and no separate exponential-in-$k$ factor is needed on
top of the $2^{O(n)}$ term. This shared worst-case bound does not mean the
two algorithms would behave identically in practice: Algorithm~5 eliminates one
Rabin state at a time and branches over the (up to $r$) Rabin states in $R$
at each call, so a recursive branch may require up to $r$ elimination rounds, whereas Algorithm~6 eliminates an entire right-hand side $R_j$ per round and branches over the (at most $k$) Rabin pairs of $\rho$, so a branch requires at most $k$ rounds. 

While $k$ could, in theory, be exponentially larger than $r$, every deletion of a Rabin set $R_j$ deletes at least one state, so there can be at most $r$ such deletion steps along an execution branch.
Furthermore, $k$ is logarithmic in $r$ for Rabin automata arising from LTL-to-Rabin translations \cite{KMSZ18}, so Algorithm~6 is expected to explore substantially fewer distinct subproblems along any given execution than Algorithm~5. 

\begin{table}[!h]
\centering
\small
\renewcommand{\arraystretch}{1.35}
\begin{tabular}{lcc}
\hline
\textbf{Algorithm} & \textbf{Time complexity} & \textbf{Main dependence} \\
\hline

Piterman--Pnueli Explicit Ranking \cite{PitermanP06} &
$O(mn^{k+1}kk!)$ &
factorial in $k$ \\

Piterman--Pnueli Symbolic Fixpoint \cite{PitermanP06} &
$O(n^{k+1}k!)$ &
factorial in $k$ \\

Colourful Universal Trees \cite{MST24} &
$\widetilde{O}\!\left(mn(k!)^{1+o(1)}\right)$ &
near-factorial in $k$ \\

 Recursive Muller Solver \cite{LKX26} & $O(|V|!\,|V|\,(|E|+k|V|))$ & factorial in $|V|$ \\ 
 
 Muller Dynamic Programming \cite{LKX26} & $O\!\left((k|V|+2^{|V|}|E|)\,|V|\right)$ & exponential in $|V|$ \\

Rabin-Vardi-Di Stasio  &
$2^{O(n)}\cdot poly(n,M,k)$ &
exponential in $n$, no factorial term \\

Vardi-Di Stasio  &
$2^{O(n)}\cdot poly(n,M,k)$ &
exponential in $n$, no factorial term \\


\hline
\end{tabular}
\caption{Complexity comparison for Rabin-game algorithms. Here $n = |V|$ is the number of states, $m=|E|$ is the number of edges, $k$ is the Rabin index, and $M=|G|$ is the size of the game representation.}
\label{tab:rabin-complexity}
\end{table}

\vspace{0.1cm}
It is shown in \cite{casares2024simple} that under the Exponential Time Hypothesis (ETH), solving a Rabin game cannot be done in time $2^{o(k log k)}\cdot n^{O(1)}$. That is, the complexity cannot be both polynomial in the Rabin number $r$ and exponential in the Rabin index $k$. As shown in Table~\ref{tab:rabin-complexity},
our new algorithms are the first to be exponential in both $n$ and $k$.